\documentclass[11pt]{article}

\usepackage{amsmath,amssymb,amsthm,mathtools}
\mathtoolsset{mathic=true}

\usepackage[T1]{fontenc}
\usepackage{libertinus}
\usepackage{libertinust1math}
\usepackage[narrow,varqu,varl,scaled=0.95]{zi4}

\usepackage[margin=1in]{geometry}

\usepackage{csquotes}
\usepackage[USenglish]{babel}
\usepackage{microtype}
\usepackage{enumitem}
\usepackage[only,llbracket,rrbracket]{stmaryrd}
\usepackage{tikz}
\usetikzlibrary{arrows.meta,positioning}
\usepackage[hidelinks,hypertexnames=false]{hyperref}
\usepackage[nameinlink,capitalize,noabbrev]{cleveref}
\usepackage[giveninits=true,maxbibnames=99,style=alphabetic,maxalphanames=4,minalphanames=3,isbn=false,maxcitenames=99,backref=true]{biblatex}
\AtEveryBibitem{%
  \clearlist{language}%
  \clearlist{location}
}
\AtBeginBibliography{%
  \setlength{\emergencystretch}{2em}%
}
\DefineBibliographyStrings{english}{
  backrefpage  = {page},
  backrefpages = {pages},
}

\AtEveryBibitem{ %
    \clearfield{day}
    \clearfield{month}
    \clearfield{series}
    \clearfield{venue}
    \clearname{editor}
    \clearlist{publisher}
    \clearlist{location} %
    \clearfield{venue}
    \clearfield{issn}
    \clearfield{isbn}
    \clearfield{urldate}
    \clearfield{eventdate}
    \clearfield{pages}
    \clearfield{number}
    \clearfield{volume}
}

\usepackage{xurl}
\usepackage[dvipsnames]{xcolor}
\usepackage{ninecolors}
\NineColors{saturation=high}

\hypersetup{
  colorlinks=true,
  linkcolor=blue2,
  citecolor=blue2,
  urlcolor=blue2,
  hypertexnames=false
}

\hypersetup{
  pdftitle={QMA has perfect completeness},
  pdfauthor={Sabee Grewal and Dorian Rudolph}
}
\allowdisplaybreaks
\setlist{
    itemsep=2pt plus .5pt minus .5pt,
    topsep=4pt plus 1pt minus 1pt,
    parsep=0pt,
    partopsep=0pt,
    labelindent=.5em,
    leftmargin=*
  }

\usepackage{aliascnt}

\newaliascnt{lemma}{theorem}
\newtheorem{lemma}[lemma]{Lemma}
\aliascntresetthe{lemma}
\newaliascnt{corollary}{theorem}
\newtheorem{corollary}[corollary]{Corollary}
\aliascntresetthe{corollary}
\newaliascnt{definition}{theorem}
\newtheorem{definition}[definition]{Definition}
\aliascntresetthe{definition}
\theoremstyle{definition}
\newaliascnt{remark}{theorem}

\aliascntresetthe{remark}
\crefname{theorem}{Theorem}{Theorems}
\crefname{lemma}{Lemma}{Lemmas}
\crefname{corollary}{Corollary}{Corollaries}
\crefname{definition}{Definition}{Definitions}
\crefname{remark}{Remark}{Remarks}
\usepackage{thm-restate}

\usepackage[
  font=small,
  labelfont=bf,
  labelsep=period
]{caption}

\newcommand{\QIP}{\mathsf{QIP}}
\newcommand{\preciseQMA}{\mathsf{PreciseQMA}}
\newcommand{\QMA}{\mathsf{QMA}}
\newcommand{\QMAone}{\mathsf{QMA}_1}
\newcommand{\QCMA}{\mathsf{QCMA}}

\newcommand{\MA}{\mathsf{MA}}
\newcommand{\BQP}{\mathsf{BQP}}

\newcommand{\Z}{{\mathbb{Z}}}
\newcommand{\calG}{{\mathcal{G}}}
\newcommand{\calR}{{\mathcal{R}}}
\newcommand{\Gtwo}{\mathcal{G}_{2}}
\newcommand{\poly}{\mathrm{poly}}
\newcommand{\Tr}{\operatorname{Tr}}
\newcommand{\ket}[1]{\lvert #1\rangle}
\newcommand{\bra}[1]{\langle #1\rvert}

\newcommand{\ketbra}[2]{\lvert #1\rangle\!\langle #2\rvert}
\newcommand{\proj}[1]{\lvert #1\rangle\!\langle #1\rvert}
\newcommand{\norm}[1]{\lVert #1\rVert}
\newcommand{\CNOT}{\operatorname{CNOT}}
\newcommand{\Toffoli}{\operatorname{Toffoli}}
\newcommand{\acc}{\mathrm{acc}}
\newcommand{\rej}{\mathrm{rej}}
\newcommand{\yes}{\mathrm{yes}}
\newcommand{\no}{\mathrm{no}}
\newcommand{\lmax}{\lambda_{\max}}
\newcommand{\calO}{\mathcal{O}}
\newcommand{\regX}{\mathcal{X}}
\newcommand{\regY}{\mathcal{Y}}
\newcommand{\regZ}{\mathcal{Z}}
\newcommand{\regJ}{\mathcal{J}}
\newcommand{\regC}{\mathcal{C}}

\usetikzlibrary{quantikz2}

\begin{document}
\title{$\QMA$ has perfect completeness}
\author{%
  Sabee Grewal\thanks{IBM Research and Columbia University. \href{mailto:sabee@ibm.com}{\texttt{sabee@ibm.com}}}
  \and
  Dorian Rudolph\thanks{Paderborn University. \href{mailto:mail@dorianrudolph.com}{\texttt{mail@dorianrudolph.com}}}
}
\date{}
\maketitle

\begin{abstract}
We prove $\QMA = \QMAone$, i.e., every quantum Merlin--Arthur proof system can be made perfectly complete. 
Our construction uses only Hadamard, Toffoli, and \(X\) gates, yielding a universal gate set for \(\QMAone\).
As a consequence, quantum \(3\)-SAT is \(\QMA\)-complete.
The construction relativizes to classical oracles, so known classical-oracle separations of \(\QMA\) from \(\QCMA\) extend to \(\QMAone\).
\end{abstract}

\section{Introduction}\label{sec:intro}

In the standard Merlin--Arthur setup, an all-powerful prover Merlin sends a proof to an efficient verifier Arthur, who must decide whether to accept or reject an input.
Such a proof system has completeness \(c\) and soundness \(s<c\) if every YES instance admits a proof that Arthur accepts with probability at least \(c\), while every purported proof on a NO instance is accepted with probability at most \(s\).
Thus, when \(c<1\), even an honest prover may be rejected with nonzero probability.
A natural question is whether, without loss of computational power, these proof systems can be made \emph{perfectly complete}---that is, whether one can take \(c=1\), so that every YES instance admits a proof that the verifier accepts with certainty.

This work concerns the complexity class \(\QMA\) (Quantum Merlin--Arthur), in which Merlin sends a polynomial-size quantum proof to an efficient quantum verifier Arthur. 
Its perfectly complete variant \(\QMAone\) was introduced by Bravyi in his study of quantum \(k\)-SAT~\cite{Bra06}.
In quantum $k$-SAT, one is given $k$-local projectors and asked whether they have a common zero-energy state, or whether every state has energy at least an inverse polynomial.
The problem naturally has perfect completeness, because, in a YES instance, a satisfying state violates none of the constraints and can therefore be accepted with certainty.
Bravyi showed that quantum $k$-SAT is $\QMAone$-complete for \(k\ge4\), and Gosset and Nagaj later proved the same for \(k=3\)~\cite{GN13}.
Whether the requirement of perfect completeness changes the computational power of \(\QMA\)---that is, whether \(\QMA=\QMAone\)---has remained open since Bravyi introduced $\QMAone$ twenty years ago. 

This state of affairs is particularly striking because perfect completeness is known for several closely related proof systems.
Both \(\MA\), where the proof is classical and the verifier is randomized, and \(\QCMA\), where the proof is classical and the verifier is quantum, are unchanged by requiring perfect completeness~\cite{ZF87,GZ11,JKNN12}.
The same is true of \(\QIP\), which allows multiple rounds of quantum interaction between the prover and verifier~\cite{KW00,KLN15}, and of \(\preciseQMA\), the variant of \(\QMA\) with an inverse-exponentially small completeness--soundness gap~\cite{fefferman_et_al:LIPIcs.ITCS.2018.4}.

A barrier to proving \(\QMA=\QMAone\) was discovered shortly after Bravyi introduced the class.
Aaronson~\cite{Aar09} constructed a quantum oracle relative to which \(\QMA\ne\QMAone\), showing that any proof of equality must exploit structure of \(\QMA\) verifiers that is unavailable in the black-box setting.
More recently, Aaronson, Harris, and Witteveen~\cite{AHW25} made this barrier quantitative by showing that black-box amplification using polynomial resources cannot, in general, reduce the completeness error below doubly exponentially small.
This matches the black-box amplification of Jeffery and Witteveen~\cite{JW26}, which achieves doubly exponentially small completeness error.

Despite these barriers, we resolve the \(\QMA\) versus \(\QMAone\) question affirmatively.
More strongly, perfect completeness can be achieved using a single fixed gate set with rational matrix entries.
Consider the gate set
\[
\Gtwo=\{X,\CNOT,\Toffoli,H\otimes H\}
\]
studied by Amy et al.~\cite{AGKMMR24}, and let \(\QMAone^{\Gtwo}\) denote the class \(\QMAone\) with the verifier restricted to gates from \(\Gtwo\).

\begin{restatable}{theorem}{mainthm}\label{thm}
\(\QMA=\QMAone^{\Gtwo}\).
\end{restatable}

The fixed-gate-set strengthening is itself significant.
For bounded-error classes such as \(\BQP\) and \(\QMA\), approximate circuit synthesis makes the choice of universal gate set essentially irrelevant.
Perfect completeness, however, leaves no room for approximation error, since approximating a gate can turn an acceptance probability of one into something strictly smaller than one.
It was therefore not previously known whether \(\QMAone\) admits a universal gate set at all~\cite{Rud25}.
\Cref{thm} settles this question as well.
For every fixed finite gate set \(\calG\), $\QMAone^{\calG}\subseteq\QMA=\QMAone^{\Gtwo}$.
Thus, \(\Gtwo\) is universal for \(\QMAone\).

\subsection{Further consequences}\label{sec:consequences}

\cref{thm} has several immediate consequences.
First, any gate set that implements \(\Gtwo\) exactly also gives perfectly complete verifiers for all of \(\QMA\).

\begin{corollary}\label{cor:gatesets}
Let \(\calG\) be a gate set that implements every gate of \(\Gtwo\) exactly, possibly using ancillas that start and end in \(\ket0\).
Then
\[
  \QMAone^{\calG}=\QMA.
\]
\end{corollary}

The gate set consisting of Hadamard, Toffoli, and \(X\) gates is therefore also universal for \(\QMAone\).
Indeed, two Hadamard gates implement \(H\otimes H\), while a \(\CNOT\) can be implemented using a Toffoli gate with one control fixed to \(\ket1\).
In fact, pairing the Hadamards is not essential to our proof.
Even with ordinary Hadamard gates, the acceptance operator remains dyadic because each contribution from the verifier is paired with one from its inverse in \(V^\dagger\Pi_{\acc}V\).
We state the result using \(\Gtwo\) because its gates have rational matrix entries, making the arithmetic structure used in the proof particularly transparent and connecting directly to the gate set studied by Amy et al.~\cite{AGKMMR24}.

\Cref{cor:gatesets} also applies to the Clifford-cyclotomic gate sets \(\mathcal G_{2^k}\) studied by Amy et al.~\cite[Theorem~1]{AGKMMR24}, which include the Clifford+\(T\) gate set.
Gosset and Nagaj~\cite{GN13} proved that quantum \(3\)-SAT is complete for \(\QMAone\) with Clifford+\(T\) verifiers.
Since \cref{cor:gatesets} gives $\QMAone^{\mathrm{Clifford}+T}=\QMA$, their result immediately yields the following.

\begin{corollary}\label{cor:3sat}
Quantum \(3\)-SAT is \(\QMA\)-complete.
\end{corollary}

Ma and Natarajan~\cite{ma_et_al:LIPIcs.ITCS.2026.101} recently introduced $XZ$-quantum $k$-SAT, a variant of quantum $k$-SAT in which each local projector is diagonal in either the standard or Hadamard basis. 
They show that $XZ$-quantum $6$-SAT is complete for $\QMAone^{\Gtwo}$. Combining with \cref{thm} implies $\QMA$-completeness. 

\begin{corollary}\label{cor:xzsat}
$XZ$-quantum \(6\)-SAT is \(\QMA\)-complete.
\end{corollary}

Another consequence of \cref{thm} concerns clique homology.
Given a graph, its clique complex is obtained by filling each clique with a simplex, and clique homology asks whether this complex has nontrivial homology in a specified dimension.
Crichigno and Kohler~\cite{CK24} first showed that the problem is \(\QMAone\)-hard, while King and Kohler~\cite{KK24} placed the gapped variant for weighted graphs in \(\QMA\).
Subsequently, Rudolph~\cite[Theorem~3.11]{Rud25} proved that the weighted gapped problem is complete for \(\QMAone^{\Gtwo}\), and Hayakawa~\cite{Hay26} obtained the same result for unweighted graphs. Together with \cref{thm}, this gives the following.

\begin{corollary}\label{cor:homology}
Gapped clique homology is \(\QMA\)-complete for both weighted and unweighted graphs.
\end{corollary}

Finally, although our proof necessarily fails to relativize to arbitrary quantum oracles, it does preserve access to classical ones.
For a classical oracle \(O\colon\{0,1\}^*\to\{0,1\}\), let \(Q_O\) denote the usual query gate
\[
  Q_O\ket{y,b}=\ket{y,b\oplus O(y)}.
\]
This gate is a permutation matrix and is its own inverse. 
As we explain below, these properties preserve the arithmetic structure that our proof relies on.

\begin{restatable}{corollary}{relativizationcor}\label{cor:relativize}
For every classical oracle \(O\) under standard XOR access, $\QMA^O=\QMAone^{\Gtwo,O}$.
The same holds for in-place permutation oracles.
\end{restatable}

Recent works of Bostanci, Haferkamp, Nirkhe, and Zhandry~\cite{BHNZ26} and Bostanci, Huang, and Vaikuntanathan~\cite{BHV26} constructed classical oracles separating \(\QMA\) from \(\QCMA\).
By \cref{cor:relativize}, the same oracles separate \(\QMAone\) from \(\QCMA\).

\begin{corollary}\label{cor:oracle}
There is a classical oracle \(O\) such that $\QCMA^O\subsetneq\QMAone^{\Gtwo,O}$.
\end{corollary}
This strengthens earlier classical-oracle separations between \(\QMAone\) and \(\QCMA\) for restricted verifier and oracle models~\cite{MPR26}.

\subsection{Proof overview}

The containment $\QMA_1^{\Gtwo} \subseteq \QMA$ is immediate, so the technical part of this work is proving the reverse inclusion.
Let $E$ be the acceptance operator of a $\QMA$ verifier on $m$ proof qubits, so that a proof $\rho$ is accepted with probability $\Tr(E \rho)$.
In the YES case, $\lambda_{\rm max}(E) \ge \tfrac{2}{3}$, whereas, in the NO case, $\norm{E} \le \tfrac{1}{3}$.
Our goal is to transform this verifier into one for which some proof is accepted with probability exactly one in the YES case, while every proof is accepted with probability bounded away from one in the NO case.

A natural approach would be for Merlin to provide a classical description of \(\lambda=\lambda_{\max}(E)\), together with a corresponding eigenvector \(\ket{\psi}\).
If \(\lambda\) had a succinct exact description, Arthur could check that \(\lambda\ge2/3\) and then certify
\[
  (\lambda I-E)\ket{\psi}=0.
\]

The difficulty is that the maximum eigenvalue need not have a simple exact representation that the verifier can use.
Even when every entry of \(E\) is a dyadic rational, \(\lambda_{\max}(E)\) can have algebraic degree exponential in \(m\).\footnote{A \emph{dyadic rational} is a rational number of the form \(a/2^h\), for integers \(a\) and \(h\ge0\). The \emph{algebraic degree} of a number \(\alpha\) is the degree of its minimal polynomial over \(\mathbb{Q}\), i.e., the lowest-degree monic polynomial in \(\mathbb{Q}[x]\) having \(\alpha\) as a root.}
Thus, the standard exact representation of an algebraic number through its minimal polynomial may already involve an object of exponential degree, and there is no general polynomial-size exact representation of \(\lambda_{\max}(E)\) that we can directly use in the verification procedure.
This is in contrast with \(\QCMA\), where the acceptance probability associated with a fixed classical proof can be arranged to have a succinct exact description~\cite{JKNN12}.
For \(\QMA\), directly certifying the optimal acceptance probability therefore appears to be the wrong approach.

The key idea is to avoid describing an eigenvalue altogether.
We first compile the $\QMA$ verifier into the gate set \(\Gtwo\), while preserving its completeness and soundness. That is, we show that $\QMA = \QMA^{\Gtwo}(2/3, 1/3)$.
For the resulting verifier, we identify a family of rational quantities \(\{\Gamma_z\}_{z\in\{0,1\}^m}\) that retain enough information about the spectrum to witness a YES instance, while admitting succinct exact descriptions.
For every YES instance, some \(\Gamma_z\) is large, and we use its exact value to construct a quantum history whose consistency can be verified with certainty.

Set $A=2E$, where $E$ is the acceptance operator of the $\QMA$ verifier.
This rescaling puts the YES and NO cases on opposite sides of \(1\).
On a YES instance, $\lambda_{\max}(A)\ge\tfrac43$,
while on a NO instance, $\norm A\le\tfrac23$.
Thus powers of \(A\) grow along some direction in the YES case, while \(A\) remains a strict contraction in the NO case.
Now take \(\ell=\Theta(m)\).
On a YES instance, we have 
\[
\Tr(A^\ell)\ge\lambda_{\max}(A)^\ell\ge(4/3)^\ell.
\]
Since $\Tr(A^\ell)  =   \sum_{z\in\{0,1\}^m}\bra zA^\ell\ket z$,
some computational-basis state \(z\) satisfies
\[
  \Gamma_z\coloneqq \bra zA^\ell\ket z
  \ge
  2^{-m}(4/3)^\ell.
\]
By choosing \(\ell=O(m)\) sufficiently large, we have $\Gamma_z>2$.

Crucially, \(\Gamma_z\) is much easier to describe exactly than \(\lambda_{\max}(E)\).
Because we compiled the $\QMA$ verifier into \(\Gtwo\), we can argue that there is a polynomially bounded integer \(h\) such that \(2^hE\) is an integer matrix.
It follows that $2^{h\ell}\Gamma_z\in\Z$, and \(\Gamma_z\) has a rational description using only \(O(h\ell)\) bits.
Thus \(\Gamma_z\), unlike the maximum eigenvalue, is an exact quantity that Merlin can describe succinctly.

The verification proceeds as follows.
Merlin supplies \(z\), an exact description of \(\Gamma_z\), and the normalized history state proportional to
\[
  \ket{\Omega_z}
  =
  \sum_{t=0}^{\ell}\ket t\otimes A^t\ket z,
\]
where successive terms are related by multiplication by \(A\).

We emphasize that the powers of \(A\) should not be interpreted as physical time evolution.
The operator \(A=2E\) is generally nonunitary and does not preserve norms.
Indeed, writing \(\ket z\) in an eigenbasis of \(A\), each eigenspace component of \(A^t\ket z\) changes geometrically with \(t\).
Consequently, after normalization, the clock distribution of \(\ket{\Omega_z}\) can be highly nonuniform and may place substantially more weight on later clock values.
This differs from the usual Feynman--Kitaev circuit history, which is a uniform superposition over a sequence of states related by unitary gates.
The nonuniformity causes no difficulty here because Arthur does not verify the history through separate local propagation checks.
Instead, he checks a single global consistency condition for the entire history.

To construct this condition, observe that the sequence \(A^t\ket z\) can be closed into a cycle by mapping the final term back to the first.
This is precisely where the claimed values \(z\) and \(\Gamma_z\) enter the verification.
Consider the scaled rank-one projection $\Gamma_z^{-1}\proj z$. 
It is easy to check that
\[
  \Gamma_z^{-1}\proj zA^\ell\ket z=\ket z.
\]
Writing \(\beta=\Gamma_z^{-1}\), we can therefore define the cyclic operator $C_{z, \beta}$ by 
\[
  C_{z,\beta}
  =
  \sum_{t=0}^{\ell-1}\ketbra{t+1}{t}\otimes A
  +
  \ketbra0\ell\otimes\beta\proj z.
\]
For the honest value \(\beta=\Gamma_z^{-1}\), the history state is an exact fixed point of this operator: $C_{z,\beta}\ket{\Omega_z}
  =
  \ket{\Omega_z}$.
Equivalently,
\[
  (I-C_{z,\beta})\ket{\Omega_z}=0.
\]

As long as Arthur can efficiently and exactly check whether Merlin's proof is in the kernel of $I- C_{z,\beta}$, we have that Arthur rejects honest history states with probability exactly zero, as desired. 
We explain in our proof how Arthur can perform such a check, and our check is similar to the kernel test in \cite{Rud25}. 
In particular, Arthur's rejection probability is proportional to 
\[
  \norm{(I-C_{z,\beta})\ket{\psi}}_2^2.
\]

The soundness of our verification makes essential use of the fact that $\norm{A} < 1$ in the NO case. 
Arthur first rejects any claimed value \(\Gamma<2\), so for every remaining claim we have \(\beta=\Gamma^{-1}\le\tfrac12\).
On a NO instance, \(\norm A\le\tfrac23\).
Since \(C_{z,\beta}\) applies either \(A\) or \(\beta\proj z\) on each clock sector, it follows that
\[
  \norm{C_{z,\beta}}
  =
  \max\{\norm A,\beta\}
  \le \frac23.
\]
Thus \(C_{z,\beta}\) is a strict contraction for every \(z\) and every claimed \(\Gamma\ge2\), regardless of whether \(\Gamma\) is equal to the true value \(\bra zA^\ell\ket z\).
Consequently, for every normalized state \(\ket\psi\),
\[
  \norm{(I-C_{z,\beta})\ket\psi}
  \ge 1-\norm{C_{z,\beta}}
  \ge \frac13.
\]
Hence no purported history state on a NO instance can satisfy the cyclic consistency condition, and the kernel test rejects with constant probability.

Putting the two cases together, a YES instance admits a choice of \(z\), \(\Gamma\), and history state for which the cyclic constraint is satisfied exactly, while on a NO instance every claimed \(z\), \(\Gamma\), and history state violates the constraint by a constant amount.
It is useful to emphasize that this is a global consistency check.
Once \(z\) and \(\Gamma\) are fixed, they determine the operator \(C_{z,\beta}\).
On a YES instance, the honest choice gives \(C_{z,\beta}\) an eigenvalue equal to \(1\), with eigenvector \(\ket{\Omega_z}\).
On a NO instance, by contrast, \(\norm{C_{z,\beta}}\le 2/3\) for every allowed \(z\) and \(\Gamma\).
Thus soundness holds against arbitrary purported history states, regardless of how their norm is distributed across the clock register.
The remaining technical work is to implement this global kernel test exactly over \(\Gtwo\), which we do in \cref{sec:verifier}.

Our construction is closely related in spirit to the recent work of Jeffery and Witteveen~\cite{JW26}.
Both constructions use history states whose amplitudes evolve geometrically.
Jeffery and Witteveen attach an infinite counter to an eigenvector of the acceptance operator and choose geometrically weighted amplitudes so that contributions from neighboring counter values cancel exactly.
The infinite counter is what makes this cancellation exact; truncating it leaves a boundary error, leading to their doubly exponentially small completeness error.

Our history state has a similar structure, with successive terms given by \(A^t\ket z\), so that each eigenspace component changes geometrically with \(t\).
The key new ingredient is to certify not the maximum acceptance probability, but instead a large diagonal entry \(\Gamma_z=\bra zA^\ell\ket z\) of a power of the acceptance operator.
Because \(\Gamma_z\) has a succinct exact rational description, Merlin can send its value to Arthur, who uses it to map the final term of the finite history exactly back to the first.
Thus, rather than using an infinite sequence to avoid a boundary, we close a finite sequence into a cycle.

This also explains how our proof avoids the black-box barriers discussed above~\cite{Aar09,AHW25}.
The succinct description of \(\Gamma_z\) relies on the arithmetic structure of the verifier: after compiling into \(\Gtwo\), the acceptance operator has dyadic rational entries with a known common denominator.
An arbitrary quantum oracle need not satisfy any such restriction.
Thus our construction does not relativize to quantum oracles, consistent with the known oracle barriers.

\subsection{Open problems}\label{sec:open}

Our results leave analogous questions about perfect completeness and universal gate sets open for other quantum complexity classes.
Two particularly natural cases are \(\BQP\) and \(\QMA(2)\).
In particular, it remains open whether \(\BQP=\BQP_1\) (equivalently, \(\BQP=\mathsf{coRQP}\)) and whether \(\QMA(2)=\QMA_1(2)\).

Our proof does not seem to extend to either setting.
For \(\BQP\), there is no prover who can supply the exact rational value used in the cyclic consistency test, and it is unclear how the verifier could compute this value itself.
For \(\QMA(2)\), soundness only bounds acceptance probabilities on product states across the two proof registers, and therefore does not yield the operator-norm bound used in our argument.

Independently of these equalities, it remains open whether \(\BQP_1\) or \(\QMA_1(2)\) admits a universal finite gate set.
That is, does either class have a fixed gate set whose perfectly complete verifiers capture every promise problem recognizable using any other gate set?

\section{Preliminaries}\label{sec:prelim}

Unless stated otherwise, gate sets are finite, act on a constant number of qubits, and have entries approximable to \(b\) bits in time \(\poly(b)\).
A \emph{verifier} is a circuit \(V\) on a proof register \(\regX\) of \(m\) qubits and an ancilla register \(\regY\) initialized to \(\ket0\), with an output qubit whose value \(1\) means acceptance.
With \(\Pi_{\acc}\) the projector onto output \(1\), the \emph{acceptance operator} is
\begin{equation}\label{eq:acceptance}
  E=(I_{\regX}\otimes\bra0_{\regY})\,V^\dagger\Pi_{\acc}V\,(I_{\regX}\otimes\ket0_{\regY}),\qquad0\preceq E\preceq I,
\end{equation}
and a proof \(\rho\) is accepted with probability \(\Tr(E\rho)\).

\begin{samepage}
\begin{definition}\label{def:qma}
  Let \(\calG\) be a gate set and \(0\le s<c\le1\).
  A promise problem \(L=(L_{\yes},L_{\no})\) is in \(\QMA^{\calG}(c,s)\) if there is a polynomial-time algorithm that maps every input \(x\) to a verifier \(V_x\) over \(\calG\), with acceptance operator \(E_x\), such that
  \begin{enumerate}[label=\textup{(\roman*)}]
    \item if \(x\in L_{\yes}\), then \(\lmax(E_x)\ge c\), that is, some proof is accepted with probability at least \(c\);
    \item if \(x\in L_{\no}\), then \(E_x\preceq sI\), that is, every proof is accepted with probability at most \(s\).
  \end{enumerate}
  We define \(\QMA=\QMA^{\calG}(2/3,1/3)\) for any universal gate set \(\calG\), and \(\QMAone^{\calG}=\QMA^{\calG}(1,1/2)\) for every gate set \(\calG\).
\end{definition}
\end{samepage}

The class \(\QMA\) is well defined: for universal gate sets, \(\QMA^{\calG}(c,s)\) does not depend on \(\calG\) or on the constants \(0<s<c<1\), by the Solovay--Kitaev theorem~\cite{Kit97,DN05,BG21} and amplification~\cite{KSV02,MW05}.
For the same reason, \(\QMA^{\calG}(c,s)\subseteq\QMA\) for every gate set \(\calG\) and all constants \(0<s<c<1\), and in particular \(\QMAone^{\calG}\subseteq\QMA\).
The perfectly complete classes are indexed by the gate set, because the argument does not apply when \(c=1\).
For gate sets implementing \(\Gtwo\) exactly, any constant soundness in \((0,1)\) gives the same perfectly complete class by standard parallel repetition: running \(r\) copies and accepting only if all accept preserves perfect completeness and reduces soundness from \(s\) to \(s^r\).
Classical reversible arithmetic on polynomially many bits, with workspace returned to zero, is available over \(\{X,\CNOT,\Toffoli\}\)~\cite{Ben73}.

\section{Proof of perfect completeness}\label{sec:proof}

We prove \cref{thm} in three steps.
First, in \cref{sec:dyadic}, we compile the original \(\QMA\) verifier into the gate set \(\Gtwo\).
The resulting acceptance operator \(E\) has a controlled dyadic denominator: there is an integer \(h=\poly(n)\) such that \(2^hE\) is an integer matrix.
This will imply that the quantities $\Gamma_z=\bra z(2E)^\ell\ket z$ have succinct exact descriptions.
Next, in \cref{sec:cycle}, we use a claimed value of \(\Gamma_z\) to construct a cyclic consistency operator \(C_{z,\Gamma}\).
On YES instances, the honest history state lies exactly in the kernel of \(I-C_{z,\Gamma}\), while on NO instances \(I-C_{z,\Gamma}\) maps every normalized state to a vector of norm at least \(1/3\).
Finally, in \cref{sec:verifier}, we show how to test this constraint exactly over \(\Gtwo\) using a block encoding of \(I-C_{z,\Gamma}\).

\subsection{Dyadic verifiers}\label{sec:dyadic}

We begin by compiling every $\QMA$ verifier into the gate set $\Gtwo$.

\begin{lemma}\label{lem:G2}
  \(\QMA=\QMA^{\Gtwo}(2/3,1/3)\).
\end{lemma}

\begin{proof}
  We follow the real simulation of Aharonov~\cite{Aha03}, with the only change that the added qubit becomes part of the proof register.
  Start from a \(\QMA\) verifier with completeness \(2/3\) and soundness \(1/3\) over the universal gate set \(\{H,\Lambda(P)\}\), where \(P=\operatorname{diag}(1,i)\) and \(\Lambda(P)\) is the controlled-\(P\) gate~\cite{Kit97,Aha03}.
  Replace each gate \(U\) by its real representation
  \begin{equation}\label{eq:real}
    \calR(U)=\begin{pmatrix}\operatorname{Re}U&-\operatorname{Im}U\\\operatorname{Im}U&\operatorname{Re}U\end{pmatrix},
  \end{equation}
  where the block index is one additional qubit \(c\) that is added to the proof register.
  The map \(\calR\) is multiplicative and preserves adjoints, and the output projector and the projection onto \(\ket0_{\regY}\) are real.
  Hence the new acceptance operator is \(\calR(E)\).
  If \(Ev=\lambda v\), then \((\operatorname{Re}v,\operatorname{Im}v)\) and \((-\operatorname{Im}v,\operatorname{Re}v)\) are eigenvectors of \(\calR(E)\) with eigenvalue \(\lambda\).
  An orthonormal eigenbasis of \(E\) gives an orthonormal eigenbasis of \(\calR(E)\) in this way, so the eigenvalues have doubled multiplicity and completeness and soundness are unchanged.
  The gate \(\calR(H)=I_c \otimes H\) is a Hadamard gate, and \(\calR(\Lambda(P))\) has entries in \(\{0,\pm1\}\), so it has an exact circuit over \(\Gtwo\)~\cite[Theorem~1]{AGKMMR24}.
  Finally, replace each Hadamard gate by \(H\otimes H\) with the second factor on one ancilla qubit.
  The ancilla is ignored, so the acceptance operator is unchanged, and the circuit is over \(\Gtwo\).
\end{proof}

\begin{corollary}\label{cor:dyadic}
  If \(V\) is a verifier over \(\Gtwo\) with \(g\) gates \(H\otimes H\), then \(2^{2g}E\) is an integer matrix.
\end{corollary}

\begin{proof}
  Every gate other than \(H\otimes H\) is a permutation matrix and \(H\otimes H\) has entries \(\pm\frac12\), so \(2^gV\) is an integer matrix.
  By \eqref{eq:acceptance}, \(E=W^\dagger\Pi_{\acc}W\) with \(W=V(I\otimes\ket0_{\regY})\).
  Since \(2^gW\) is an integer matrix, so is \(2^{2g}E\).
\end{proof}

The point of the preceding compilation is to obtain exact arithmetic control over the acceptance operator.
From now on, let \(V\) be the verifier of \cref{lem:G2} for a fixed input, with \(m\) proof qubits and acceptance operator \(E\).
If \(g\ge1\) is the number of \(H\otimes H\) gates in \(V\), then, by \cref{cor:dyadic}, $2^hE\in\Z^{2^m\times2^m}$ where $h=2g$.
Thus \(E\) has a known dyadic denominator of polynomial size, which will allow Merlin to describe certain entries of powers of \(E\) exactly using only polynomially many bits.

Finally, let \(A=2E\).
Then \(A\succeq0\), and the original completeness--soundness gap becomes
\[
  \lambda_{\max}(A)\ge\frac43
  \quad\text{on YES instances},\qquad
  \norm A\le\frac23
  \quad\text{on NO instances}.
\]
This rescaling will let us amplify the two cases in opposite directions by taking powers of \(A\).

\subsection{A cyclic constraint}\label{sec:cycle}

We now construct the cyclic consistency constraint that will underlie the perfectly complete verifier.
We begin by using the arithmetic structure of $E$ to construct the classical part of Merlin's proof. 

Let \(\ell\) be the least integer with \(\ell\ge4(m+1)\) such that \(\ell+1\) is a power of two, so \(\ell<8(m+1)\).
The \emph{clock register} \(\regC\) has \(\log_2(\ell+1)\) qubits with basis \(\ket0,\dots,\ket\ell\), and \(S\ket t=\ket{t+1\bmod(\ell+1)}\) is the cyclic shift on \(\regC\).

\begin{lemma}\label{lem:entry}
  For every \(z\in\{0,1\}^m\), the number \(\Gamma_z=\bra zA^\ell\ket z\) satisfies \(0\le\Gamma_z\le2^\ell\) and \(2^{h\ell}\Gamma_z\in\Z\).
  In particular, \(\Gamma_z\) is a rational number whose numerator and denominator have at most \((h+1)\ell=\calO(h\ell)\) bits.
  On a YES instance, some \(z\) has \(\Gamma_z>2\).
\end{lemma}

\begin{proof}
  \(2^{h\ell}A^\ell=2^\ell(2^hE)^\ell\) is an integer matrix, and \(0\le\Gamma_z\le\norm A^\ell\le2^\ell\).
  On a YES instance,
  \(\max_z\Gamma_z\ge2^{-m}\Tr(A^\ell)\ge2^{-m}(4/3)^{4(m+1)}>2^{-m}\cdot2^{m+1}=2\), since \((4/3)^4>2\).
\end{proof}

\Cref{lem:entry} gives the exact scalar that will be used to construct the witness.
Consider the sequence $\ket z, A\ket z,A^2\ket z,\dots,A^\ell\ket z$.
Each vector is obtained from the previous one by applying $A$.
The key observation is that $\Gamma_z$ allows the final vector to be returned to the first.
Indeed,
\[
 \Gamma_z^{-1}\proj zA^\ell\ket z = \Gamma_z^{-1}\bra zA^\ell\ket z\,\ket z = \ket z.
\]
Thus the sequence can be closed into a cycle.

In the $\QMAone$ verification procedure, Merlin supplies $z$ together with a claimed value $\Gamma$ of $\Gamma_z$. 
Arthur uses this claimed value to define the closing transition and hence the cyclic operator in the following lemma. 
The resulting construction is illustrated in \cref{fig:cycle}.

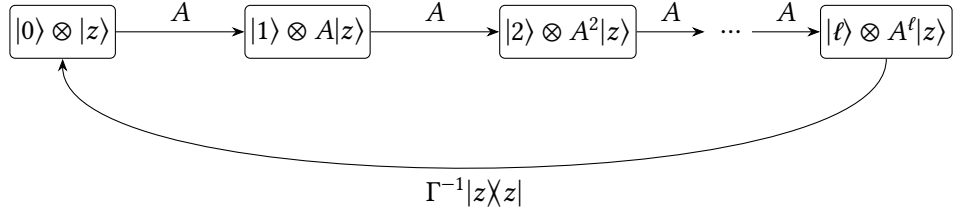
\begin{figure}[t]
  \centering
  \begin{tikzpicture}[>=Stealth,node distance=17mm,
      st/.style={draw,rounded corners=2pt,minimum width=13mm,minimum height=7mm,inner sep=2pt}]
    \node[st] (t0) {\(\ket0\otimes\ket z\)};
    \node[st,right=of t0] (t1) {\(\ket1\otimes A\ket z\)};
    \node[st,right=of t1] (t2) {\(\ket2\otimes A^2\ket z\)};
    \node[right=9mm of t2] (dots) {\(\vphantom{+}\cdots\)};
    \node[st,right=9mm of dots] (tl) {\(\ket\ell\otimes A^\ell\ket z\)};
    \draw[->] (t0) -- node[above] {\(A\)} (t1);
    \draw[->] (t1) -- node[above] {\(A\)} (t2);
    \draw[->] (t2) -- node[above] {\(A\)} (dots);
    \draw[->] (dots) -- node[above] {\(A\)} (tl);
    \draw[->] (tl.south) to[out=-90,in=-90,looseness=0.45] node[below] {\(\Gamma^{-1}\proj z\)} (t0.south);
  \end{tikzpicture}
  \caption{The cyclic constraint.
  Each arrow is one transition of \(C_{z,\Gamma}\); the clock advances by one modulo \(\ell+1\) at every transition.
  The history state \(\ket{\Omega_z}\) is the sum of the boxed components and is a fixed point of \(C_{z,\Gamma}\) when \(\Gamma=\Gamma_z\).
  On a NO instance every arrow is a contraction of norm at most \(2/3\).}
  \label{fig:cycle}
\end{figure}

\begin{lemma}\label{lem:kernel-gap}
  For \(z\in\{0,1\}^m\) and \(\Gamma\ge2\), define the operator \(C_{z,\Gamma}\) and history vector \(\ket{\Omega_z}\) on \(\regC\otimes\regX\) by
  \begin{equation}\label{eq:C}
    \begin{aligned}
      C_{z,\Gamma}&=\sum_{t=0}^{\ell-1}\ketbra{t+1}{t}\otimes A+\ketbra0\ell\otimes\Gamma^{-1}\proj z,\\
      \ket{\Omega_z}&=\sum_{t=0}^{\ell}\ket t\otimes A^t\ket z.
    \end{aligned}
  \end{equation}
  \begin{enumerate}[label=\textup{(\roman*)}]
    \item On a YES instance, if \(\Gamma_z>2\) and \(\Gamma=\Gamma_z\), then
    \(\ket{\Omega_z}\ne0\) satisfies \((I-C_{z,\Gamma})\ket{\Omega_z}=0\).
    \item On a NO instance, \(\norm{(I-C_{z,\Gamma})\ket\Psi}\ge\frac13\norm{\ket\Psi}\) for all \(z\), all \(\Gamma\ge2\), and all \(\ket\Psi\).
  \end{enumerate}
\end{lemma}

\begin{proof}
  (i) Since \(\Gamma^{-1}\proj zA^\ell\ket z=\ket z\), the operator \(C_{z,\Gamma}\) cyclically permutes the summands of \(\ket{\Omega_z}\), leaving their sum fixed.
  (ii) Factor
  \begin{equation}\label{eq:C-factor}
    C_{z,\Gamma}=(S\otimes I)\Bigl(\sum_{t=0}^{\ell-1}\proj t\otimes A+\proj\ell\otimes\Gamma^{-1}\proj z\Bigr).
  \end{equation}
  The first factor is unitary and the second is block diagonal with blocks \(A\) and \(\Gamma^{-1}\proj z\), so \(\norm{C_{z,\Gamma}}=\max\{\norm A,\Gamma^{-1}\}\le2/3\).
  Hence \(\norm{(I-C_{z,\Gamma})\ket\Psi}\ge\norm{\ket\Psi}-\norm{C_{z,\Gamma}\ket\Psi}\ge\frac13\norm{\ket\Psi}\).
\end{proof}

\subsection{The \texorpdfstring{\(\QMAone\)}{QMA1} verifier}\label{sec:verifier}

By \cref{lem:kernel-gap}, it remains to test whether the purported history state lies in the kernel of \(I-C_{z,\Gamma}\).
We implement this test using a block encoding of \(I-C_{z,\Gamma}\).
The block encoding is constructed so that a designated measurement outcome occurs with probability proportional to $\norm{(I-C_{z,\Gamma})\ket\Psi}^2$.
Arthur rejects on this outcome.
Thus, the honest history on a YES instance is rejected with probability exactly zero, while \cref{lem:kernel-gap}(ii) will give a constant rejection probability on NO instances.

To construct this block encoding, it is useful to separate the two types of transitions in the cycle.
Define
\begin{equation}\label{eq:transitions}
T=\sum_{t=0}^{\ell-1}\ketbra{t+1}{t}\otimes E,
\qquad
R_z=\ketbra0\ell\otimes\proj z.
\end{equation}
Here, \(T\) gives the forward transitions in the history, while \(R_z\) gives the final transition back to clock state \(\ket 0\).
Since \(A=2E\), we have
\begin{equation}\label{eq:IC-decomposition}
I-C_{z,\Gamma} = I-2T-\Gamma^{-1}R_z.
\end{equation}

There are two issues in implementing \eqref{eq:IC-decomposition} exactly.
First, we need exact block encodings of \(T\) and \(R_z\).
Second, we need to incorporate the rational coefficient \(\Gamma^{-1}\) using only gates from \(\Gtwo\).
We first explain the second point, assuming block encodings of \(T\) and \(R_z\).
Write the claimed value as \(\Gamma=a/b\), so that \(\Gamma^{-1}=b/a\).
Introduce a \(k\)-qubit index register with \(M=2^k\) basis states.
We prepare this register in the uniform superposition $\frac{1}{\sqrt M}\sum_{j=0}^{M-1}\ket j$,
and, conditioned on the value of \(j\), apply one of several operations.
For \(a\) values of \(j\), the relevant block is \(I\); for \(2a\) values it is \(-T\); and for \(b\) values it is \(-R_z\).
For all remaining values of \(j\), the relevant block is zero.

After applying \(H^{\otimes k}\) again and projecting the index register onto \(\ket0\), these contributions are averaged uniformly.
The resulting block is therefore
\[
  \frac1M\bigl(aI-2aT-bR_z\bigr)
  =
  \frac aM\bigl(I-2T-(b/a)R_z\bigr)
  =
  \frac aM\bigl(I-C_{z,a/b}\bigr).
\]
Thus the coefficient \(\Gamma^{-1}=b/a\) is realized exactly by the relative numbers of index values assigned to the different operations.

To construct block encodings of \(T\) and \(R_z\), we use the ancilla register \(\regY\) of \(V\), initialized to \(\ket0\), and a shared flag qubit \(f\), initialized to \(\ket1\).
In each case, the desired block is obtained by projecting \(\regY\) and \(f\) onto \(\ket0\) at the output.
Let \(W_T\) apply \(V\), flip \(f\) if the clock is below \(\ell\) and the output qubit of \(V\) is \(1\), and then apply \(V^\dagger\) followed by the cyclic clock shift \(S\).
Let \(W_{R_z}\) flip \(f\) if the history register is in state \(\ket\ell\otimes\ket z\), and then apply \(S\).
\Cref{fig:WT} shows these two gadgets.

\begin{figure}[htb]
  \centering
  \begin{minipage}[t]{0.60\textwidth}
    \centering
    \vspace{0pt}
    \begin{quantikz}[wire types={b,b,b,q}, classical gap=0.06cm, row sep={1cm, between origins}, column sep=0.5cm]
      \lstick{clock \(\regC\)} & & \gate[style={rounded corners=5pt}]{I-\proj\ell}\wire[d][1]{q} & & \gate{S} & \\
      \lstick{proof \(\regX\)} & \gate[2][0.85cm]{V} & \gate[2,style={rounded corners=5pt}]{\Pi_{\acc}}\wire[d][2]{q} & \gate[2][0.85cm]{V^\dagger} & & \\
      \lstick{ancilla \(\regY:\ket0\)} & & & & & \\
      \lstick{flag \(f:\ket1\)} & & \targ{} & & &
    \end{quantikz}
    \par\smallskip
    {\small (a) \(W_T\): forward transition $\ket{t}\to\ket{t+1}$}
  \end{minipage}\hfill
  \begin{minipage}[t]{0.38\textwidth}
    \centering
    \vspace{0pt}
    \begin{quantikz}[wire types={b,b,b,q}, classical gap=0.06cm, row sep={1cm, between origins}, column sep=0.5cm]
      & \gate[style={rounded corners=5pt}]{\proj\ell}\wire[d][1]{q} & \gate{S} & \\
      & \gate[style={rounded corners=5pt}]{\proj z}\wire[d][2]{q} & & \\
      & & & \\
      & \targ{} & &
    \end{quantikz}
    \par\smallskip
    {\small (b) \(W_{R_z}\): closing transition $\ket{\ell}\to\ket{0}$}
  \end{minipage}
  \caption{Block encodings of \(T\) and \({R_z}\).
   Rounded boxes denote controls on the indicated projector subspaces.}
  \label{fig:WT}
\end{figure}
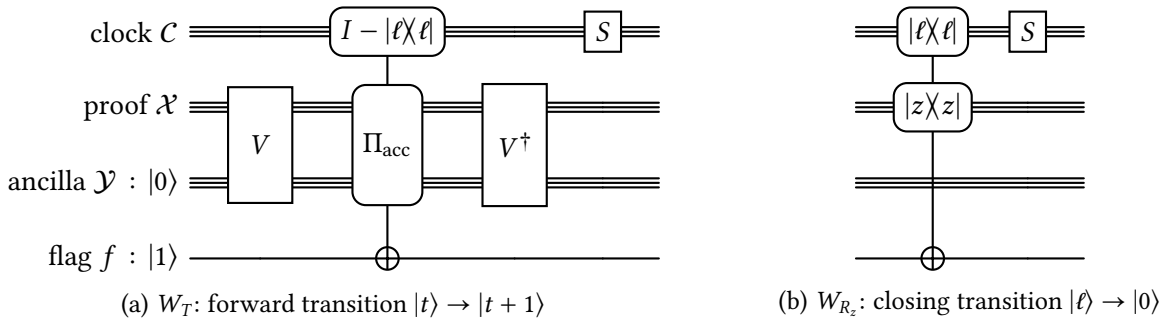

\begin{definition}\label{def:U}
  For \(z\in\{0,1\}^m\) and nonnegative integers \(a,b\) with \(3a+b\le M\), define
  \begin{equation}\label{eq:U}
    U_{z,a,b}=(H^{\otimes k}\otimes I)\Bigl(\sum_{j=0}^{M-1}\proj j_{\regJ}\otimes U_j\Bigr)(H^{\otimes k}\otimes I),
    \qquad
    U_j=\begin{cases}
      X_f&j<a,\\
      -W_T&a\le j<3a,\\
      -W_{R_z}&3a\le j<3a+b,\\
      I&j\ge3a+b.
    \end{cases}
  \end{equation}
  Its ancilla register is \(\regZ=\regJ\regY f\), where \(\regJ\) has \(k\) qubits, and is initialized to \(\ket\eta_{\regZ}=\ket0_{\regJ}\ket0_{\regY}\ket1_f\).
\end{definition}

\begin{figure}[htb]
  \centering
  \begin{quantikz}[wire types={b,b,b,b,q}, classical gap=0.06cm, row sep={0.85cm,between origins}, column sep=0.4cm]
    \lstick{index \(\regJ:\ket0\)} & \gate{H^{\otimes k}} & \gate[style={rounded corners=5pt}]{j<a}\wire[d][4]{q} & \gate[style={rounded corners=5pt}]{a\le j<3a}\wire[d][1]{q} & \gate[style={rounded corners=5pt}]{3a\le j<3a+b}\wire[d][1]{q} & \gate{H^{\otimes k}} & \\[0.2cm]
    \lstick{clock \(\regC\)} & & & \gate[4]{-W_T} & \gate[4]{-W_{R_z}} & & \\
    \lstick{proof \(\regX\)} & & & & & & \\
    \lstick{ancilla \(\regY:\ket0\)} & & & & & & \\
    \lstick{flag \(f:\ket1\)} & & \targ{} & & & &
  \end{quantikz}
  \caption{The construction of \(U_{z,a,b}\) from the gadgets in \cref{fig:WT}.
  Between the two Hadamard layers, the index register coherently selects \(X_f\), \(-W_T\), or \(-W_{R_z}\).
  For \(j\ge3a+b\), all controls are inactive, giving the identity branch.}
  \label{fig:U}
\end{figure}
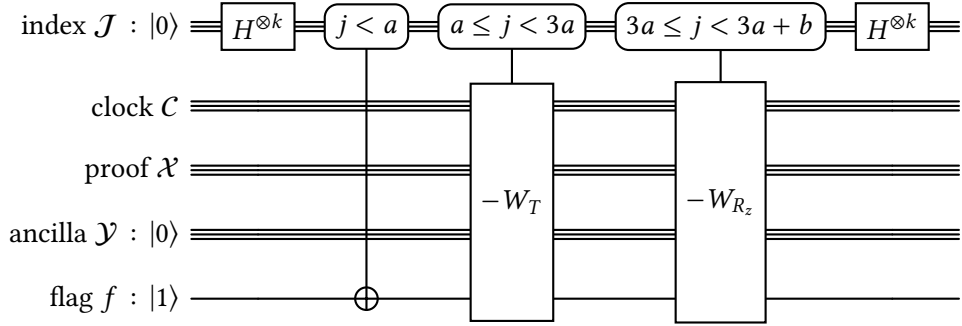

For the verifier, choose \(k\) to be the least even integer with \(k\ge(h+1)\ell+4\).
A \emph{valid certificate} is a triple \((z,a,b)\) with \(z\in\{0,1\}^m\) and positive integers \(a,b\) representable by $k$ bits satisfying \(M/8\le a\le M/4\) and \(2b\le a\).
These conditions ensure that the index ranges fit, \(3a+b\le\frac72a<M\), that \(\Gamma=a/b\ge2\) as required by \cref{lem:kernel-gap}, and that the prefactor \(a/M\) is at least \(1/8\).

\begin{lemma}\label{lem:block}
  For every valid certificate, \(U_{z,a,b}\) satisfies
  \[
    (\bra0_{\regZ}\otimes I)\,U_{z,a,b}\,(\ket\eta_{\regZ}\otimes I)=\frac aM\bigl(I-C_{z,a/b}\bigr).
  \]
  It has a polynomial-size circuit over \(\Gtwo\), computable in polynomial time from \(V\) and the certificate, that applies \(V\) and \(V^\dagger\) once each and uses additional workspace qubits that start and end in \(\ket0\).
\end{lemma}

\begin{proof}
  A flag flip controlled on a projector \(P\) has block \(P\) between flag input \(\ket1\) and output \(\ket0\).
  Together with \cref{eq:acceptance}, this gives the transition blocks
  \[
    \begin{aligned}
      (\bra0_{\regY}\bra0_f)W_T(\ket0_{\regY}\ket1_f)
        &=(S\otimes I)((I-\proj\ell)\otimes E)=T,\\
      (\bra0_{\regY}\bra0_f)W_{R_z}(\ket0_{\regY}\ket1_f)
        &=(S\otimes I)(\proj\ell\otimes\proj z)=R_z.
    \end{aligned}
  \]
  For \(j<a\), \(X_f\) flips the flag from \(\ket1\) to \(\ket0\), giving block \(I\); for \(j\ge3a+b\), the identity leaves the flag in \(\ket1\), giving block \(0\).

  Since \(H^{\otimes k}\ket0=M^{-1/2}\sum_j\ket j\), we have \((\bra0_{\regJ}\otimes I)U_{z,a,b}(\ket0_{\regJ}\otimes I)=M^{-1}\sum_jU_j\).
  Between input \(\ket0_{\regY}\ket1_f\) and output \(\ket0_{\regY}\ket0_f\), the four branches have blocks \(I\), \(-T\), \(-R_z\), and \(0\), so the claimed block is
  \[
    \frac1M(aI-2aT-bR_z)=\frac aM(I-C_{z,a/b})
  \]
  by \cref{eq:IC-decomposition}.

  We implement the operation between the two Hadamard layers in \cref{fig:U} as follows.
  First apply \(V\) to \(\regX\regY\).
  Flip \(f\) if \(a\le j<3a\), the clock is below \(\ell\), and the output qubit of \(V\) is \(1\).
  Then apply \(V^\dagger\).
  Both \(V\) and \(V^\dagger\) are applied independently of \(j\); outside \(a\le j<3a\), the flag flip is disabled, so they cancel.
  Next, flip \(f\) if \(j<a\), or if \(3a\le j<3a+b\) and the history register is in state \(\ket\ell\otimes\ket z\).
  Finally, for \(a\le j<3a+b\), apply the clock shift \(S\) and the phase flip \(\ket j_{\regJ}\mapsto-\ket j_{\regJ}\), giving the minus signs in \(-W_T\) and \(-W_{R_z}\).
  Since \(k\) is even, each \(H^{\otimes k}\) consists of \(k/2\) gates \(H\otimes H\).
  Every gate of \(\Gtwo\) is its own inverse, so \(V^\dagger\) is the reversed gate list of \(V\).
  All branch and flag conditions and the controlled clock shift use polynomial-size reversible circuits over \(\{X,\CNOT,\Toffoli\}\)~\cite{Ben73}, with workspace returned to zero.
  The only additional gate is \(Z=\operatorname{diag}(1,-1)\), applied to a qubit storing the range predicate \(a\le j<3a+b\); it has an exact implementation over \(\Gtwo\)~\cite{AGKMMR24}.
  These predicates use arithmetic on \(k\)-bit integers, so their circuit size is polynomial in \(k=\log_2 M\).
\end{proof}

With \(\regZ\) initialized to \(\ket\eta\), the verifier applies \(U_{z,a,b}\) and rejects exactly when measuring \(\regZ\) gives all zeros.
By \cref{lem:block}, a normalized history \(\ket\psi\) has rejection probability
\begin{equation}\label{eq:rejection}
  p_{\rej}=\left(\frac aM\right)^2\norm{(I-C_{z,a/b})\ket\psi}^2.
\end{equation}

Thus a kernel vector is accepted with certainty. 

\begin{samepage}
\begin{lemma}\label{lem:certificate}
  On a YES instance, there is a valid certificate with \(a/b=\Gamma_z>2\).
\end{lemma}

\begin{proof}
  We rescale the numerator and denominator by the same power of two so that \(a\) is a constant fraction of \(M\).
  Take \(z\) with \(\Gamma_z>2\) from \cref{lem:entry}, \(b_0=2^{h\ell}\), and \(a_0=b_0\Gamma_z\in\Z\), so \(2b_0<a_0\le2^{(h+1)\ell}\).
  Let \(u=\lceil\log_2a_0\rceil\le(h+1)\ell\) and \(s=k-u-2>0\).
  Then \(a=2^sa_0\) and \(b=2^sb_0\) satisfy \(a/M=a_0/2^{u+2}\in(\frac18,\frac14]\), \(2b<a\), and \(a/b=\Gamma_z\).
\end{proof}
\end{samepage}

\mainthm*

\begin{proof}
  Let \(L\in\QMA\), and let \(V=V_x\) be the verifier of \cref{lem:G2}.
  The new verifier expects a certificate \((z,a,b)\) in a computational-basis register and a history register \(\regC\otimes\regX\).
  It measures the certificate and rejects if it is invalid.
  Otherwise it initializes \(\regZ\) to \(\ket\eta\), applies \(U_{z,a,b}\) to the history and ancilla registers, measures \(\regZ\), and rejects if and only if the outcome is all zero.
  The measured values \(z,a,b\) are used as classical parameters in the comparisons defining \(U_{z,a,b}\).

  On a YES instance, the prover sends the certificate of \cref{lem:certificate} and the normalized history \(\ket{\Omega_z}/\norm{\ket{\Omega_z}}\) from \eqref{eq:C}.
  By \cref{lem:kernel-gap}(i), the rejection probability is exactly zero, so the verifier accepts with probability one.
  On a NO instance, for every valid certificate, \eqref{eq:rejection}, \cref{lem:kernel-gap}(ii), and \(a/M\ge1/8\) give rejection probability at least \(\frac1{64}\cdot\frac19=\frac1{576}\) for every state of the history register.
  Averaging over certificate outcomes, every proof is therefore accepted with probability at most \(1-\frac1{576}\).

  By \cref{lem:block}, the result is a polynomial-size verifier over \(\Gtwo\), so \(L\in\QMA^{\Gtwo}(1,1-\frac1{576})=\QMAone^{\Gtwo}\).
\end{proof}

We next show that the equality $\QMA=\QMA_1^{\calG_2}$ relativizes to classical oracles.

\relativizationcor*

\begin{proof}
  For standard XOR oracles, the query gate is a real permutation matrix and is its own inverse.
  The real simulation in \cref{lem:G2} maps \(Q_O\) to \(I \otimes Q_O\), leaving query gates unchanged.
  Query gates introduce no denominators, so \cref{cor:dyadic} still holds, and their self-inverseness allows \(V^\dagger\) to be implemented by reversing the gate list.
  The block encoding applies \(V\) and \(V^\dagger\) without additional controls, so the rest of the construction is unchanged.

  For in-place permutation oracles, the same construction works if inverse queries are temporarily allowed.
  Controlled oracle queries are available by~\cite[full version, Lemma~2.7 and Remark~2.8]{MPR26}.
  To remove inverse queries, let the prover supply the standard Feynman--Kitaev circuit history of the verifier~\cite{KSV02}.
  An inverse step has the form \(\ket\alpha\ket t\mapsto\ket\beta\ket{t+1}=(O^\dagger\ket\alpha)\ket{t+1}\).
  We can check this relation using a forward query instead: apply \(O\) conditioned on the clock being \(t+1\), then perform the identity-gate propagation check between clock times \(t\) and \(t+1\).
  This gives exactly the usual propagation penalty, since
  \(\norm{\ket\alpha-O\ket\beta}^2=\norm{O^\dagger\ket\alpha-\ket\beta}^2.\)
  With this observation, the proof of~\cite[full version, Lemma~3.11]{MPR26} gives inverse queries for free.
  Here the history state records the individual gates of the verifier, including oracle queries and their inverses; the history in \eqref{eq:C} records successive powers of \(A\).
  The additional non-oracle gates have exact implementations over \(\Gtwo\)~\cite{AGKMMR24}.
\end{proof}

\pdfbookmark[1]{Acknowledgments}{acknowledgments}
\paragraph*{Acknowledgments.}
SG is supported by the Herman Goldstine Memorial Postdoctoral Fellowship.
DR is supported by the Deutsche Forschungsgemeinschaft (DFG) under project 563388236 (Bridge-QS, SPP 2514).

\pdfbookmark[1]{Tool and computational resource disclosure}{tool-disclosure}
\paragraph*{Tool and computational resource disclosure.}
Generative AI was used extensively in the development of this paper, including in exploring and refining proof ideas and in drafting and revising the text. In particular, the proof idea underlying the main theorem is the result of prompting a generative AI system, which proposed the core argument developed in this paper. The authors subsequently verified, simplified, and organized the argument, and edited its exposition. They take full responsibility for the correctness and content of the paper. This disclosure is intended solely to describe the role of generative AI in the research process and should not be read as an endorsement of such systems or of the companies that produce them.

\printbibliography

\end{document}